\documentclass[superscriptaddress,reprint,showpacs,aps,pra]{revtex4-1}
\usepackage{graphics}
\usepackage{bm}
\usepackage{bbm}
\usepackage{amsmath}
\usepackage{amssymb}
\usepackage{graphicx}
\usepackage{amsthm}
\usepackage{comment}

\newtheorem{theorem}{Theorem}

\newtheorem{corol}{Corollary}

\newtheorem{prop}{Proposition}

\begin{document}

\title{One-body entanglement as a quantum resource in fermionic systems}
\author{N.\ Gigena}
\affiliation{IFLP/CONICET and Departamento  de F\'{\i}sica,
    Universidad Nacional de La Plata, C.C. 67, La Plata (1900), Argentina}
\author{M.\ Di Tullio}
\affiliation{IFLP/CONICET and Departamento  de F\'{\i}sica,
    Universidad Nacional de La Plata, C.C. 67, La Plata (1900), Argentina}
\author{R.\ Rossignoli}
\affiliation{IFLP/CONICET and Departamento  de F\'{\i}sica,
    Universidad Nacional de La Plata, C.C. 67, La Plata (1900), Argentina}
\affiliation{Comisi\'on de Investigaciones Cient\'{\i}ficas (CIC), La Plata (1900), Argentina}

\begin{abstract}
We show that one-body entanglement, which is a measure of the deviation of a pure fermionic state from a Slater determinant (SD) and is determined by the mixedness of the single-particle density matrix (SPDM), can be considered as a quantum resource. The associated theory has SDs and their convex hull as free states, and number conserving fermion linear optics operations (FLO), which include one-body unitary transformations and measurements of the occupancy of single-particle modes, as the basic free operations. We first provide a bipartitelike formulation of one-body entanglement, based on a Schmidt-like decomposition of a pure $N$-fermion state, from which the SPDM [together with the $(N-1)$-body density matrix] can be derived. It is then proved that under FLO operations, the initial and postmeasurement SPDMs always satisfy a majorization relation, which ensures that these operations cannot increase, on average, the one-body entanglement. It is finally shown that this resource is consistent with a model of fermionic quantum computation which requires correlations beyond antisymmetrization. More general free measurements and the relation with mode entanglement are also discussed.
\end{abstract}
\maketitle

\section{Introduction}

Quantum entanglement and identical particles are two fundamental concepts in quantum mechanics. Entanglement in systems of distinguishable components is particularly valuable in the field of quantum information theory \cite{NC.00} because it can be considered as a resource within the Local Operations and Classical Communication (LOCC) paradigm \cite{NC.00,PV.07}. Extending the notion of entanglement to the realm of indistinguishable particles  is, however, not  straightforward because the constituents of the system cannot be individually accessed. Different approaches have been considered, like mode entanglement \cite{Za.02,Shi.03,FL.13}, where subsystems correspond to a set of single-particle (SP) states in a given basis, extensions based on correlations between observables \cite{BK.04, ZL.04,SI.11, BG.13,BF.14} and entanglement beyond symmetrization \cite{SC.01,SL.01,ES.02, GM.02, PY.01, WV.03,IV.13,OK.13,SL.141, GR.15,MB.16}, which is independent of the choice of SP basis. Several studies on the relation between these types of entanglement  \cite{FL.13,WV.03,GR.15,DD.16,GR.17,BF.17,DGR.18,DRGC.19,SD.19,DS.20} and on whether exchange correlations can be associated with entanglement  \cite{CM.07,KC.14,CC.18,FC.18,MY.19} have been recently made. There is also a growing interest in quantum chemistry simulations based on optical lattices \cite{DDL.03,AC.19}, which would benefit from a detailed characterization of fermionic correlations. In this paper we will focus on entanglement beyond antisymmetrization in fermionic systems and analyze its consideration as a quantum resource. 

Quantum resource theories \cite{CG.19,PRXM.18} have recently become  a topic of great interest since they essentially describe quantum information processing under a restricted set of operations. Standard entanglement theory in systems of  distinguishable components is just one of these theories, amongst which we may include others like  quantum thermodynamics \cite{GH.16, GM.15}, coherence  \cite{BCP.14,SA.17}, nonlocality \cite{BC.14} 
and non-Gaussianity \cite{WP.12}. 

In the usual entanglement theory a multipartite quantum system shared by distant parties is considered. These parties can operate each on their own subsystem and  are allowed to communicate via classical channels \cite{PV.07}. From these restrictions the LOCC set arises naturally as the set of free operations of the resource theory, and  the set of free (separable) states is then derived. In our case ignoring antisymmetrization correlations defines Slater determinants (SDs) and their convex hull as the set  of ``free'' states $\cal S$ and we are looking for a set of free operations $\cal O$ consistent with this set. 

With this aim, we first define a partial order relation on the Fock space $\cal F$ of the system, based  on the mixedness of the corresponding {\it single-particle density matrix} (SPDM) $\rho^{(1)}$ [also denoted as the one-particle or one-body density matrix (DM)], which determines whether a given pure fermionic state can be considered more entangled than another state. A bipartite like formulation for this {\it one-body entanglement}, involving $\rho^{(1)}$ and the $(N-1)$-body density matrix (isospectral for pure states of $N$ fermions) is also provided.  Next we define a class of operations consistent with $\cal S$ {\it and} the previous partial order, through a majorization relation to be fulfilled by the initial and final SPDMs, which ensures that one-body entanglement will not be increased by such operations. We then show that number conserving Fermion linear optics (FLO) operations \cite{TD.02, Br.05, TB.19},  which  include one-body unitary transformations and  measurement of the  occupancy of a SP state, are indeed within this class. One-body entanglement then plays  the role of a resource in a theory where $\cal S$ is the convex hull of  SDs and $\cal O$ is that of FLO operations. Possible extensions of the set of free operations and connection of this resource with a quantum computation model and with mode entanglement are also discussed. 
 
\section{Formalism}
 
\subsection{One-body entanglement} 

We consider a SP space ${\cal H}$ of finite  dimension $n$ and a set of fermion creation and annihilation operators $c^\dag_k$ and $c_k$   associated with  an   orthogonal basis  of ${\cal H}$, satisfying the anticommutation relations $\{c_k,c_{k'}^\dag\}=\delta_{kk'}$, $\{c_k,c_{k'}\}=\{c^\dag_k,c^\dag_{k'}\}=0$. The elements of  the SPDM $\rho^{(1)}$ in a general  fermionic state $\rho$ are given by 
\begin{equation}
\rho^{(1)}_{kk'}=\langle c^\dag_{k'} c_k\rangle={\rm Tr}\,\rho\, c^\dag_{k'} c_k.
\end{equation} 
They form a Hermitian matrix with eigenvalues $\lambda_\nu=\langle c^\dag_\nu c_\nu\rangle \in[0,1]$, where $c^\dag_\nu=\sum_k U_{k\nu}c^\dag_k$ creates a fermion in one of the ``natural'' SP orbitals diagonalizing $\rho^{(1)}$ ($\langle c^\dag_{\nu'}c_{\nu}\rangle=\lambda_\nu\delta_{\nu\nu'}$). For pure states $\rho=|\Psi\rangle\langle\Psi|$, the ``mixedness'' of  $\rho^{(1)}$ then reflects the deviation of $|\Psi\rangle$ from a SD $[\prod_\nu (c^\dag_\nu)^{n_\nu}]|0\rangle$,  since for the latter $\lambda_\nu=n_\nu=0$ or $1$ $\forall\, \nu$ and hence $(\rho^{(1)})^2=\rho^{(1)}$. 

Such mixedness can be rigorously characterized through  majorization \cite{Ni.01,NV.01,MOA.11,Bh.97}. For  states $|\Psi\rangle$ and $|\Phi\rangle$ with the same fermion number $N={\rm Tr}\,\rho^{(1)}_{\Psi}={\rm Tr}\,\rho^{(1)}_{\Phi}$, we will say that $|\Psi\rangle$ is not less {\it one-body entangled} than $|\Phi\rangle$ if $\rho^{(1)}_{\Psi}$ is {\it more (or equally) mixed} than $\rho^{(1)}_{\Phi}$,  i.e.\ if their eigenvalues $\bm\lambda=(\lambda_1,\ldots,\lambda_n)$, sorted in decreasing order,  satisfy the majorization relation 
\begin{equation}
\bm{\lambda}(\rho^{(1)}_{\Psi})\prec\bm{\lambda}(\rho^{(1)}_{\Phi})\,,\label{2}
\end{equation}
which means 
\begin{equation}
\sum_{\nu=1}^m\lambda_\nu(\rho_\Psi^{(1)})\leq \sum_{\nu=1}^m\lambda_\nu(\rho^{(1)}_\Phi)
\end{equation}
for $m=1,\ldots,n-1$, with identity for $m=n$. Thus SDs are the least one-body entangled states, as their SPDM  majorizes any other $\rho^{(1)}$ with the same trace. Relation \eqref{2} is analogous to that imposed by LOCC operations on reduced states of systems of distinguishable components, which in the bipartite case lead to the celebrated Nielsen's theorem: $|\Psi_{AB}\rangle$ can be converted by LOCC to $|\Phi_{AB}\rangle$ (and hence is not less entangled than $|\Phi_{AB}\rangle$) if and only if their reduced states satisfy $\bm{\lambda}(\rho^{A(B)}_{\Psi})\prec\bm{\lambda}(\rho^{A(B)}_{\Phi})$ \cite{N.99,NV.01}. Local measurements reduce the ignorance about the state of the measured subsystem, decreasing the mixedness of reduced states and hence bipartite entanglement. Similarly, we will show that one-body entanglement  will decrease under operations which reduce the ignorance about the SPDM.

\subsubsection{The associated Schmidt decomposition}

We first  remark that one-body entanglement also admits  a bipartite like formulation: A pure state $|\Psi\rangle$ of $N$ fermions ($\sum_k c^\dag_k c_k|\Psi\rangle=N|\Psi\rangle$) can be expanded as 
\begin{eqnarray}
|\Psi\rangle&=&\frac{1}{N}\sum_{k,l}\Lambda_{kl}c^\dag_kC^\dag_l|0\rangle
\label{SD1}
\end{eqnarray}
where $C^\dag_l=c^\dag_{l_1}\ldots c^\dag_{l_{N-1}}$, $l=1,\ldots,(^{\;\;\;n}_{N-1})$,  are operators creating $N-1$ fermions in specific SP states labeled by $l$, satisfying $\langle 0|C_l C^\dag_{l'}|0\rangle=\delta_{ll'}$, while the coefficients $\Lambda_{kl}$ form  an  $n\times (^{\;\;\;n}_{N-1})$ matrix $\Lambda$  satisfying ${\rm Tr}\,\Lambda\Lambda^\dag=N$. Thus, each term in the sum \eqref{SD1} is a SD which is repeated $N$ times, such that  
\begin{equation} 
c_k|\Psi\rangle= \sum_{l}\Lambda_{kl}C^\dag_l|0\rangle
\label{ckp}
\end{equation} 
is the (unnormalized) state of remaining fermions when SP state  $k$ is occupied, while 
\begin{equation} 
C_l|\Psi\rangle=(-1)^{N-1}\sum_k \Lambda_{kl}c^\dag_k|0\rangle\label{clp}
\end{equation} 
is that of remaining fermion when the $N-1$ SP states $l$ are occupied. In this way, $\langle \Psi|\Psi\rangle=\frac{1}{N}{\rm Tr}\,\Lambda \Lambda^\dag=1$. Moreover, Eqs.\ \eqref{ckp}--\eqref{clp} allow us to  express the elements of both the SPDM $\rho^{(1)}$ {\it and the $(N-1)$-body DM $\rho^{(N-1)}$} in terms of $\Lambda$ as 
\begin{eqnarray}
\rho^{(1)}_{kk'}&=&\langle \Psi|c^\dag_{k'} c_k|\Psi\rangle=(\Lambda\Lambda^\dag)_{kk'}\,,\label{rh1}\\
\rho^{(N-1)}_{ll'}&=&\langle\Psi| C^\dag_{l'} C_l|\Psi\rangle=(\Lambda^T\Lambda^*)_{ll'}\label{rhn1}\,.
\end{eqnarray}
Eqs.\ \eqref{rh1}--\eqref{rhn1} are analogous to those for the reduced states $\rho^{A(B)}$ of distinguishable subsystems in a standard pure  bipartite state $|\Psi_{AB}\rangle=\sum_{i,j}C_{ij}|i_A,j_B\rangle$, where $\rho^{A}_{ii'}=\langle |i'_A\rangle\langle i_A|\rangle=(CC^\dag)_{ii'}$, $\rho^{B}_{jj'}=\langle |j'_B\rangle\langle j_B|\rangle=(C^{_T}C^*)_{jj'}$ \cite{NC.00}. The only difference is that ${\rm Tr}\,\rho_{A(B)}={\rm Tr}\,CC^\dag=1$ whereas ${\rm Tr}\,\rho^{(1)}={\rm Tr}\,\rho^{(N-1)}=N$. 

Eqs.\ \eqref{rh1}--\eqref{rhn1} imply that $\rho^{(1)}$ {\it and $\rho^{(N-1)}$} have {\it the same nonzero eigenvalues $\lambda_\nu$}, which are just the square of the singular values of $\Lambda$. Moreover, by means of the singular value decomposition $\Lambda=UDV^\dag$, with $D_{\nu\nu'}=\sqrt{\lambda_\nu}\delta_{\nu\nu'}$ and $U$ and $V$ unitary matrices (of $n\times n$ and $(^{\;\;\;n}_{N-1})\times (^{\;\;\;n}_{N-1})$ respectively), we may now obtain from \eqref{SD1} the $1$--$(N-1)$ Schmidt-like decomposition  of the $N$-fermion state: 
\begin{eqnarray}
|\Psi\rangle&=&\frac{1}{N}\sum_\nu\sqrt{\lambda_\nu} c^\dag_\nu C^\dag_\nu|0\rangle\,,   \label{SD}
\end{eqnarray}
where 
\begin{equation}
c^\dag_\nu=\sum_{k}U_{k\nu}c^\dag_k\;\;\;\;, \;\;\;\;\;C^\dag_\nu=\sum_l V_{l\nu}^*C^\dag_l
\end{equation}
are the ``natural'' one- and $N-1$- fermion creation operators satisfying  
\begin{eqnarray}
\langle 0|c_\nu  c^\dag_{\nu'}|0\rangle&=&\delta_{\nu\nu'}=\langle 0|C_\nu C^\dag_{\nu'}|0\rangle\,.\\
\langle \Psi|c^\dag_\nu  c_{\nu'}|\Psi\rangle&=&\lambda_\nu\delta_{\nu\nu'}=     \langle \Psi|C^\dag_\nu C_{\nu'}|\Psi\rangle\,.
\end{eqnarray}
Thus,
\begin{equation}
c_\nu|\Psi\rangle=\sqrt{\lambda_\nu}C^\dag_\nu|0\rangle\,, \;\;\;C_\nu|\Psi\rangle=(-1)^{N-1}\sqrt{\lambda_\nu}c^\dag_\nu|0\rangle\,, 
\end{equation}
i.e.\ the orthogonal natural $N-1$-fermion states $C^\dag_\nu|0\rangle$ are  those of remaining fermions when the natural SP orbital $\nu$ is occupied, while $c^\dag_\nu|0\rangle$ are the orthogonal states of the remaining fermion when the natural $N-1$-fermion state $C^\dag_{\nu}|0\rangle$ (which in general is no longer a SD) is occupied. Therefore, in an $N$-fermion state one-body entanglement is actually the $1$--$(N-1)$  -body entanglement, associated with the correlations between one- and $N-1$-body observables. 

In the case of a SD $|\Psi\rangle=(\prod_{\nu=1}^N  c^\dag_\nu)|0\rangle$, $\lambda_\nu=1$ ($0$) for $\nu\leq N$ $(>N)$, with $C^\dag_\nu\propto\prod_{\nu'\neq \nu}^N c^\dag_{\nu'}$ such that $c^\dag_\nu C^\dag_\nu|0\rangle=|\Psi\rangle$ for $\nu\leq N$. 
On the other hand, for $N=2$  Eq.\ \eqref{SD} becomes the Slater decomposition of a two-fermion state \cite{SC.01,SL.01,ES.02}, 
\begin{equation} 
|\Psi\rangle=\sum_{\nu}\sqrt{\lambda_\nu}c^\dag_\nu c^\dag_{\bar{\nu}}|0\rangle=\frac{1}{2}\sum_\nu\sqrt{\lambda_\nu}(c_\nu C^\dag_\nu+c_{\bar{\nu}}C^\dag_{\bar{\nu}})|0\rangle\label{sfn2}
\end{equation}
where $C^\dag_\nu=c^\dag_{\bar{\nu}}$, $C^\dag_{\bar{\nu}}=-c^\dag_\nu$. In this case one-body entanglement is directly related to that between the set of normal $\nu$ and $\bar{\nu}$ modes, which contain each just one-fermion (see sec.\ \ref{B4}). 

\subsubsection{One-body entanglement entropies}

We may now define a general one-body entanglement entropy $E(|\Psi\rangle)\equiv E^{(1)}(|\Psi\rangle)$ as
\begin{equation}
E(|\Psi\rangle)=S(\rho^{(1)}_{\Psi})=S(\rho^{(N-1)}_{\Psi})\,,\label{E}
\end{equation} 
where $S(\rho^{(1)})$ is a Schur-concave  function  \cite{MOA.11,Bh.97} of $\rho^{(1)}$. These entropies will all satisfy 
\begin{equation}
E(|\Psi\rangle)\geq E(|\Phi\rangle)\label{3}\,,
\end{equation} 
whenever the majorization relation of Eq.\ \eqref{2} is fulfilled. For instance, trace-form entropies 
\begin{equation} 
S(\rho^{(1)})={\rm Tr}f(\rho^{(1)})=\sum\limits_\nu f(\lambda_\nu)\label{seq}\end{equation} 
where $f:[0,1]\rightarrow \mathbb{R}$ is concave and satisfies $f(0)=f(1)=0$ \cite{CR.02}, will fulfill \eqref{3}, with $E(|\Psi\rangle)\geq 0$ $\forall$ $|\Psi\rangle$  and $E(|\Psi\rangle)=0$  if and only if $|\Psi\rangle$ is  a SD.  Such  $E(|\Psi\rangle)$ will then be one-body entanglement monotones. Examples are the von Neumann entropy of $\rho^{(1)}$, 
$S(\rho^{(1)})=-\sum_\nu\lambda_\nu \log_2\lambda_\nu$, a quantity of interest in various fields \cite{ER.96, GJ.97, ZG.97, ZS.99, BS.15},  and the one-body entropy \cite{GR.15,DGR.18} 
\begin{equation}
S_1(\rho^{(1)})=-\sum_\nu \lambda_\nu \log_2\lambda_\nu+(1-\lambda_\nu)\log_2(1-\lambda_\nu),\label{4}
\end{equation}
which represents, for $|\Psi\rangle$ of  definite fermion number $N$, the minimum relative entropy (in the grand canonical ensemble) between $\rho=|\Psi\rangle\langle\Psi|$ and any fermionic Gaussian state $\rho_g$: $\!S_1(\rho^{(1)}_{\Psi})=\mathop{\rm Min}_{\rho_g}S(\rho||\rho_g)$ \cite{DGR.18}, for  $S(\rho||\rho')=-{\rm Tr}\rho(\log_2\rho'-\log_2\rho)$ and $\rho_g\propto\exp[-\sum_{k,k'}\alpha_{kk'}c^\dag_k c_{k'}]$ (pair creation and annihilation terms in $\rho_g$ are not required for such $|\Psi\rangle$ \cite{DGR.18}). It is also the minimum over all SP bases of the sum of all single mode entropies \cite{GR.15} $-p_k\log_2 p_k-(1-p_k)\log_2 (1-p_k)$, where $p_k=\langle c^\dag_k c_k\rangle$. We remark that the SPDM and hence any measure \eqref{E}  are in principle experimentally accessible. Measurement of the fermionic SPDM in optical lattices has been recently reported \cite{AHE.18}. 

All measures \eqref{E} can be extended to mixed states 
\begin{equation}
\rho=\sum_\alpha p_\alpha|\Psi_\alpha\rangle\langle\Psi_\alpha|
\end{equation}
of definite $N$ through their convex roof extension $E(\rho)={\rm Min}\sum_\alpha p_\alpha E(|\Psi_\alpha\rangle)$, where the minimum is over all representations $\{p_\alpha\geq 0,\,|\Psi_\alpha\rangle\}$ of $\rho$ \cite{GR.15}. Such $E(\rho)$ represents a one-body entanglement of formation, vanishing if and only if $\rho$ is a convex mixture of SDs.

\subsection{One-body entanglement nongenerating operations}

\subsubsection{Definition and basic properties}

We now define a class of operations which do not generate one-body entanglement, i.e., which do not increase, on average, the mixedness of the SPDM.  
 
{\bf Definition 1.} {\it Let $\varepsilon(\rho)=\sum_j {\cal K}_j\rho{\cal K}_j^\dagger$ be a quantum operation on a fermion state $\rho$, with $\{{\cal K}_j,\,\sum_j {\cal K}_j^\dag{\cal K}_j=\mathbbm{1}\}$ a set of Kraus operators, assumed number conserving. Let $\rho^{(1)}$ and  $\rho^{(1)}_j$ be the SPDMs determined by $\rho$ and $\rho_j={\cal K}_j\rho {\cal K}^\dagger_j/p_j$, with $p_j={\rm Tr}[\rho {\cal K}^\dagger_j {\cal K}_j]$. We say that $\varepsilon$ is one-body entanglement nongenerating (ONG) if it admits a set of Kraus operators $\{{\cal K}_j\}$ satisfying $\forall\,\rho$ the  relation
\begin{equation}
\bm\lambda(\rho^{(1)})\prec \sum_j p_j\,  \bm\lambda(\rho_j^{(1)})\,,\label{avmaj}
\end{equation}
where eigenvalues $\bm\lambda(\rho^{(1)}_j)$ are sorted in decreasing order.} 

This majorization relation is analogous to that satisfied by reduced local states under local operations in the standard entanglement theory [48] and implies  
\begin{equation}
S(\rho^{(1)})\geq S\left[\sum_j p_j\bm{\lambda}(\rho^{(1)}_j)\right]\geq\sum_j p_j S(\rho^{(1)}_{j})\,,\label{sineq}
\end{equation}
for any concave entropy $S(\rho^{(1)})$, such as those of Eq.\ \eqref{seq}. For pure states $\rho=|\Psi\rangle\langle\Psi|$, $\rho_j=|\Phi_j\rangle\langle\Phi_j|$ is also pure $\forall j$,  with  $|\Phi_j\rangle\propto {\cal K}_j|\Psi\rangle$, and Eqs.\ \eqref{E}, \eqref{sineq} imply 
\begin{equation}
E(|\Psi\rangle)\geq \sum_j p_j E(|\Phi_j\rangle)\geq E(\varepsilon(|\Psi\rangle\langle\Psi|)) \,,\label{entrop}
\end{equation}
showing that any one-body  entanglement monotone  \eqref{E} will not increase, on average, after ONG operations. In particular, if $|\Psi\rangle$ is a SD, $E(|\Psi\rangle)=0$ and Eq.\ \eqref{entrop} implies  that all states $|\Phi_j\rangle\propto {\cal K}_j|\Psi\rangle$ must  be SDs or zero, i.e.\ {\it all  Kraus operators fulfilling \eqref{avmaj} should map free states onto free states}. And for the one-body entanglement of formation of general mixed states $\rho$, Eq.\  \eqref{entrop} implies 
\begin{equation} 
E(\rho)\geq E(\varepsilon(\rho))
\end{equation}
since by using  the minimizing representation, $E(\rho)=\sum_\alpha p_\alpha E(|\Psi_\alpha\rangle)\geq \sum_{\alpha,j}p_\alpha p_{\alpha j}E(|\Phi_{\alpha j}\rangle)\geq E(\varepsilon(\rho))$.   

It also follows from \eqref{avmaj}  that the set of ONG operations is convex and {\it closed under composition}, i.e., $\bm{\lambda}(\rho^{(1)})\prec \sum_{i,j} p_{ij}\bm{\lambda}(\rho^{(1)}_{ij})$ for ${\cal K}_{ij}={\cal K}^b_i{\cal K}^a_j$ and  $\varepsilon(\rho)=\varepsilon^b[\varepsilon^a(\rho)]$. This property ensures that ONG operations can be applied any number of times in any order. 
\begin{prop}
The conversion of a pure state $|\Psi\rangle\in{\cal F}$ into another pure state $|\Phi\rangle\in{\cal F}$ by means of ONG operations is possible only   if the majorization relation \eqref{2} is satisfied by the corresponding SPDMs. 
\end{prop}
\begin{proof}
The state conversion will consist in some sequence of ONG operations, which can be resumed in just one ONG operation due to the closedness under composition. Let $\{{\cal K}_j\}$ be a set of associated Kraus operators satisfying \eqref{avmaj}. After this operation is performed, we should have   ${\cal K}_j|\Psi\rangle=\sqrt{p_j}|\Phi\rangle \,\forall j$,  with $p_j=\langle\Psi|{\cal K}^\dagger_j {\cal K}_j|\Psi\rangle$, implying  $\rho^{(1)}_j =\rho^{(1)}_\Phi\,\forall j$ and hence Eq.\ \eqref{2} when (\ref{avmaj}) is fulfilled.
\end{proof}
Then, maximally one-body entangled states are those pure states whose SPDM is majorized by that of any other state. Due to Eq.\ \eqref{3} they will also maximize $E(|\Psi\rangle)$ for any choice of $S$. At fixed fermion number $N\geq 2$ and $n=mN$ they are states leading to   \begin{equation}
\rho^{(1)}=\mathbbm{1}_n/m\,,
\end{equation}
for which any SP basis is natural. For $m$ integer such $\rho^{(1)}$ emerges, for instance, from Greenberger-Horne-Zeilinger (GHZ)-like states involving  superpositions of SDs in orthogonal subspaces:
\begin{equation}
|\Psi\rangle=\frac{1}{\sqrt{m}}\sum_{l=0}^{m-1} c^\dag_{Nl+1}\ldots c^\dag_{Nl+N}|0\rangle=\frac{1}{N\sqrt{m}}\sum_{\nu=1}^nc^\dag_\nu C^\dag_\nu|0\rangle\,,\label{FGZ}
\end{equation}  
which lead to $\langle c^\dag_\nu c_{\nu'}\rangle=\delta_{\nu\nu'}/m$. 

\subsubsection{Fermion linear optics operations as ONG}

We now show that number conserving FLO operations \cite{TD.02, Br.05, TB.19},  which include one-body unitary transformations and measurement of the occupancy of a SP mode, are  included in the ONG set. First, any number conserving one-body unitary transformation 
\begin{equation} 
{\cal U}=\exp[-i\sum_{k,k'}H_{k'k}c^\dag_{k'}c_{k}]\,,
\end{equation} 
with ${\cal U}^\dag{\cal U}=\mathbbm{1}$ ($H^\dag=H$) is obviously ONG: Since ${\cal U}c^\dag_k{\cal U}^\dag=\sum_{k'} U_{k'k}c^\dag_{k'}$, with $U=e^{-iH}$, it will map  the SPDM as $\rho^{(1)}\rightarrow U^\dag\rho^{(1)}U$, leaving its eigenvalues unchanged (and hence transforming SDs into SDs). It  can  be implemented through composition of phaseshifting and beamsplitters  unitaries \cite{TD.02, Br.05, TB.19} ${\cal U}_{p}(\phi)=e^{-i\varphi\, c^\dagger_kc_k}$, ${\cal U}_{b}(\theta)=e^{-i\theta\,(c^\dagger_k c_{k'} + c^\dagger_{k'} c_k)}$, which are the basic unitary elements of the FLO set. 
    
FLO operations also include  measurements of the occupancy of single-particle modes, described by projectors
\begin{equation}
{\cal P}_k=c^\dagger_k c_k\,,\; {\cal P}_{\bar k}=c_k c^\dagger_k\,, \label{Pk}
\end{equation}
which satisfy ${\cal P}_k+{\cal P}_{\bar{k}}=\mathbbm{1}$. We now show explicitly the following fundamental result.  

\begin{theorem}\label{t1}
The measurement of the occupancy of a single-particle state $|k\rangle=c^\dagger_k|0\rangle\in{\cal H}$, described by the operators \eqref{Pk}, is a ONG operation. 
\end{theorem}

\begin{proof}
Consider a general pure fermionic state $|\Psi\rangle$ with SPDM $\rho^{(1)}$. Let $\rho_k^{(1)}$ and $\rho_{\bar k}^{(1)}$ be the SPDMs after SP mode $|k\rangle$ is found to be occupied or empty, respectively, determined by the states 
\begin{equation}
|\Psi_k\rangle={\cal P}_k |\Psi\rangle/\sqrt{p_k}\,, \;\;|\Psi_{\bar k}\rangle={\cal P}_{\bar k}|\Psi\rangle/\sqrt{p_{\bar k}}\,,\label{pek}
\end{equation} 
with  $p_k=\langle\Psi|{\cal P}_k|\Psi\rangle=1-p_{\bar k}$. Then 
\begin{equation}
|\Psi\rangle=\sqrt{p_k}|\Psi_k\rangle+\sqrt{p_{\bar k}}|\Psi_{\bar{k}}\rangle\,.\label{pkk}\end{equation}
We will prove relation (\ref{avmaj}), i.e. (Fig \ref{fig:1}),
\begin{equation}
\bm\lambda(\rho^{(1)})\prec p_k\bm\lambda(\rho_k^{(1)})  + p_{\bar k}\bm\lambda(\rho^{(1)}_{\bar k})\,.\label{major}
\end{equation}

If the measured state $|k\rangle$ is a natural orbital, such that $\langle c^\dag_k c_{k'}\rangle=p_k\delta_{kk'}$ with  $p_k=\lambda_k$ an eigenvalue of $\rho^{(1)}$, Eq.\ \eqref{major} is straightforward: In this case   \eqref{pkk} leads to 
\begin{equation}
\rho^{(1)}=p_k\rho^{(1)}_k+p_{\bar{k}}\rho^{(1)}_{\bar{k}}\,,
\label{rhod}
\end{equation}

\begin{figure}[h]
\includegraphics[width=8cm,trim=0cm 1.5cm 0cm 0cm;clip]{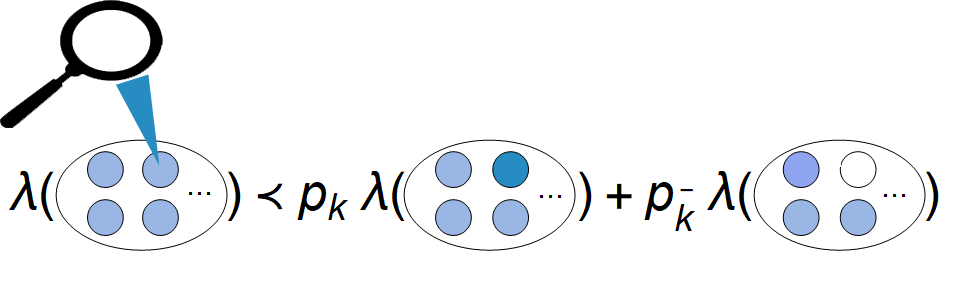}
\caption{Measurement of the occupancy of a single fermion mode $k$. It  reduces (or does not increase), on average, the mixedness of the SP density matrix  $\rho^{(1)}$ ($\bm{\lambda}$ denotes its  spectrum) and hence the one-body entanglement.}\label{fig:1}
\end{figure}

as $\langle\Psi_k|c^\dag_{k''}c_{k'} |\Psi_{\bar{k}}\rangle=\delta_{k''k}(1-\delta_{k'k})\langle c^\dag_k c_{k'}\rangle=0$ $\forall$ $k',k''$. Eq.\ \eqref{rhod}  implies \eqref{major} since $\bm{\lambda}(A+B)\prec\bm{\lambda}(A)+\bm{\lambda}(B)$ for any two  Hermitian $n\times n$ matrices $A$ and $B$ \cite{NV.01} (this case includes the trivial situation  $p_k=1$ or  $0$,   where $|\Psi\rangle=|\Psi_k\rangle$ or $|\Psi_{\bar{k}}\rangle$; in the following we consider $p_k\in(0,1)$).

Otherwise Eq.\ \eqref{rhod} no longer holds. Nevertheless, since $\langle\Psi_k|c^\dag_{k''}c_{k'} |\Psi_{\bar{k}}\rangle=0$ for any two SP states $|k'\rangle$ and $|k''\rangle$ orthogonal to $|k\rangle$,  Eq.\ \eqref{pkk} implies, for  any SP subspace ${\cal S}_{\perp}\subset{\cal H}$ orthogonal to the measured state $|k\rangle$,
\begin{equation}
\rho^{(1)}_{{\cal S}_{\perp}}=p_k\rho^{(1)}_{k{\cal S}_{\perp}}+p_{\bar{k}}\rho^{(1)}_{\bar{k}{\cal S}_{\perp}}\,,\label{rhort}
\end{equation}
where $\rho^{(1)}_{{\cal S}_\perp}=P_{{\cal S}_\perp}\rho^{(1)}P_{{\cal S}_\perp}$ and $\rho^{(1)}_{k(\bar{k}){\cal S}_\perp}$ are the restrictions of $\rho^{(1)}$ and $\rho^{(1)}_{k(\bar{k})}$ to ${\cal S}_\perp$, and $P_{{\cal S}_\perp}$ is the associated projector. This result is  expected  since the measurement   is ``external'' to  ${\cal S}_\perp$ (if $[{\cal K}_j,O]=0$  $\forall\, j$ $\Rightarrow$ ${\rm Tr}\,[\rho O]= {\rm Tr}\sum_j {\cal K}_j\rho{\cal K}_j^\dag O=\sum_j p_j{\rm Tr}[\rho_j O]$; for ${\cal K}_j={\cal P}_{k(\bar{k})}$ and $O=c^\dag_{k''}c_{k'}$ with $k'$ and $k''$ orthogonal to $k$, this result implies  Eq.\  \eqref{rhort}). And  for any ${\cal S}\subset{\cal H}$ containing the state $|k\rangle$, we have, as $\langle\Psi_k|c^\dag_{k'}c_{k'}|\Psi_{\bar{k}}\rangle=0$ for $k'=k$ or $k'$ orthogonal to $k$, 
\begin{equation}
{\rm Tr}\,\rho^{(1)}_{\cal S}={\rm Tr}\,[p_k \rho^{(1)}_{k{\cal S}}+p_{\bar{k}}\rho^{(1)}_{\bar{k}{\cal S}}]\,.\label{Tr}
\end{equation} 
We can now prove the  $m^{\rm th}$ inequality in \eqref{major}, 
\begin{equation}
\sum_{\nu=1}^m \lambda_\nu(\rho^{(1)})\leq \sum_{\nu=1}^m \;\left(p_k\lambda_\nu(\rho^{(1)}_k)+p_{\bar{k}}\lambda_\nu(\rho^{(1)}_{\bar{k}})\right)\,.\label{major1}
\end{equation}
Let  ${\cal S}_m\subset {\cal H}$ be the subspace spanned by the first  $m$  eigenstates of $\rho^{(1)}$, such that $\lambda_\nu(\rho^{(1)}_{{\cal S}_m}) =\lambda_\nu(\rho^{(1)})$ for $\nu\leq m$ and hence ${\rm Tr}\,\rho^{(1)}_{{\cal S}_m}=\sum_{\nu=1}^m \lambda_\nu(\rho^{(1)})$. If ${\cal S}_m$ is either orthogonal to $|k\rangle$ or fully contains $|k\rangle$, Eq.\ \eqref{rhort} or Eq.\ \eqref{Tr} holds for ${\cal S}={\cal S}_m$, implying 
\eqref{major1} since ${\rm Tr}\,\rho^{(1)}_{k(\bar{k}){\cal S}_m}\leq \sum_{\nu=1}^m\lambda_\nu(\rho^{(1)}_{k(\bar{k})})$ by the Ky Fan maximum principle \cite{NV.01} (the $m$ largest eigenvalues $\lambda_\nu$ of a Hermitian matrix $O$ satisfy $\sum_{\nu=1}^m\lambda_\nu \geq {\rm Tr}P'_m O=\sum_{\nu=1}^m\lambda'_\nu$ for any rank $m$ orthogonal projector $P'_m$, with $\lambda'_\nu$ the sorted eigenvalues of $P'_m OP'_m$). 

Otherwise we add to ${\cal S}_m$ the component  $|k_{\perp}\rangle$  of $|k\rangle$ orthogonal to  ${\cal S}_m$, obtaining an $m+1$ dimensional SP subspace ${\cal S}'_{m}$ where Eq.\ \eqref{Tr} holds and still $\lambda_\nu(\rho^{(1)}_{{\cal S}'_m})=\lambda_\nu(\rho^{(1)})$ for $\nu\leq m$. It is proved in the Appendix \ref{A} that  the remaining smallest eigenvalue satisfies 
\begin{equation}
\lambda_{m+1}(\rho^{(1)}_{{\cal S}'_{m}})\geq p_k\lambda_{m+1}(\rho^{(1)}_{k{\cal S}'_{m}})+p_{\bar{k}}\lambda_{m+1}(\rho^{(1)}_{\bar{k}{\cal S}'_{m}})\,. \label{ineq}
\end{equation}
Hence $\sum_{\nu=1}^m \lambda_\nu(\rho^{(1)})$ $ \leq \sum_{\nu=1}^m\, p_k\lambda_\nu(\rho^{(1)}_{k{\cal S}'_m})+p_{\bar{k}}\lambda_\nu(\rho^{(1)}_{\bar{k}{\cal S}'_m})$ 
due to the trace conservation  \eqref{Tr} for ${\cal S}={\cal S}'_m$, which implies  Eq.\ \eqref{major1} due to previous Ky Fan inequality. This completes the proof for pure states. It is easily verified (see Appendix \ref{B}) that these measurements map SDs onto SDs, as implied by Eq.\ \eqref{major}. 

The previous proof actually shows the general relation 
\begin{equation}
\bm\lambda(\rho^{(1)}_{\cal S})\prec p_k\bm\lambda(\rho_{k{\cal S}}^{(1)})  + p_{\bar k}\bm\lambda(\rho^{(1)}_{{\bar k}{\cal S}})\,, \label{majors2}
\end{equation}
valid for the restriction of $\rho^{(1)}$ to {\it any} subspace ${\cal S}\subset{\cal H}$ either containing or orthogonal to the measured  state $|k\rangle$ ($[P_{\cal S},|k\rangle\langle k|]=0$). Moreover  $\rho_{\cal S}^{(1)}$ will be determined by a mixed reduced state $\rho_{\cal S}={\rm Tr}_{{\cal S}_{\perp}}|\Psi\rangle\langle\Psi|$ satisfying $\langle \Psi|O_{\cal S}|\Psi\rangle={\rm Tr}\,\rho_{\cal  S}\,O_{\cal S}$ for any operator $O_{\cal S}$ involving just creation and annihilation of SP states $\in {\cal S}$ (see Appendix \ref{A}). Eq.\ \eqref{majors2} then shows that  \eqref{major} holds for {\it general mixed fermionic states} $\rho$ (assumed to commute with the fermion number  $\hat{N}=\sum_k c^\dag_k c_k$ or the number parity $e^{i\pi\hat{N}}$) since they can be purified and seen as a reduced state $\rho_{\cal S}$ of a pure fermionic state $|\Psi\rangle$ in an enlarged SP space (Eq.\ \eqref{psip} in Appendix \ref{A}). 
\end{proof}

\subsubsection{More general ONG measurements and operations}

By composing the basic measurements \eqref{Pk}, more complex operations  satisfying \eqref{avmaj} are obtained. In particular, a measurement in a basis of SDs,  which is obviously ONG, results from the composition of all measurements $\{P_k,P_{\bar{k}}\}$ in a given SP basis. Extension of the set of free operations beyond the standard FLO set can also be considered. The proof of theorem \ref{t1} can be extended to more general single-mode measurements: 

\begin{corol}
A general measurement on single-particle mode $k$ described by the operators \begin{equation}
{\cal M}_k=\alpha\, {\cal P}_k+\beta\,{\cal P}_{\bar{k}}\,,\;\; {\cal M}_{\bar k}=\gamma\,{\cal P}_{k}+\delta\, {\cal P}_{\bar{k}}\,, \label{Mk}
\end{equation}
where ${\cal M}_k^\dag {\cal M}_k+{\cal M}_{\bar k}^\dag {\cal M}_{\bar k}=\mathbbm{1}$  ($|\alpha|^2+|\gamma|^2= |\beta|^2+|\delta|^2=1$) is also a ONG operation.
\end{corol}
The proof is given in Appendix A and implies that Eq.\ \eqref{major} will be satisfied for $\rho_{k(\bar{k})}^{(1)}$ the SPDMs obtained from $|\Psi'_{k(\bar{k})}\rangle\propto{\cal M}_{k(\bar{k})}|\Psi\rangle$ and $p_k\rightarrow p'_k=\langle \Psi|{\cal M}_k^\dag{\cal M}_k|\Psi\rangle=1-p'_{\bar{k}}$. This result entails  that {\it any} pair of Kraus operators ${\cal M}_k$ and ${\cal M}_{\bar{k}}$ for the occupation measurement  operation  \begin{equation}
\varepsilon(\rho)={\cal P}_k\rho{\cal P}_{k}+{\cal P}_{\bar{k}}\rho{\cal P}_{\bar{k}}={\cal M}_k\rho{\cal M}_k^\dag+{\cal M}_{\bar{k}}\rho{\cal M}_{\bar{k}}^\dag\end{equation}
will also be ONG operations, since they are a special case of \eqref{Mk} [$\alpha \beta^*+\gamma\delta^*=0$, i.e.\ $(^{\alpha\;\beta}_{\gamma\;\delta})$ unitary].

It is also possible to consider in the present context operations which do not conserve the fermion number $N$ but still generate states with definite particle number when applied to such states. In this case it becomes necessary to extend the partial order \eqref{2} to states with different particle number. We then notice that Eq.\ \eqref{2} implies a similar majorization relation (see Appendix \ref{C}) 
\begin{equation}
\bm{\lambda}(D^{(1)}_\Psi)\prec \bm{\lambda}(D^{(1)}_{\Phi})\label{Dm}
\end{equation} 
for the sorted eigenvalues of the extended $2n\times 2n$ SPDM 
\begin{equation}
D^{(1)}=\rho^{(1)}\oplus(\mathbbm{1}-\rho^{(1)\,T})=\begin{pmatrix}\rho^{(1)}&0\\0&\mathbbm{1}-\rho^{(1)\,T}\end{pmatrix}\,, \label{D1}
\end{equation} 
with spectrum $(\bm{\lambda}(\rho^{(1)}),1-\bm{\lambda}(\rho^{(1)}))$ and elements $\langle c^\dag_{k'}c_k\rangle$ and $\langle c_{k'}c^\dag_{k}\rangle$, the trace ${\rm Tr}\,D^{(1)}=n$ of which is fixed by  the SP space dimension and is $N$-{\it independent}. For general states we then say that  $|\Psi\rangle$ is not less one-body entangled than $|\Phi\rangle$  if Eq.\ \eqref{Dm} holds. Note that {\it all} SDs lead to the same sorted spectrum $\bm{\lambda}(D^{(1)})$ regardless of $N$, all being then least entangled states, with $(D^{(1)})^2=D^{(1)}$ if and only if $|\Psi\rangle$ is a SD. Similarly, Eq.\ \eqref{avmaj} for number conserving ONG operations implies 
\begin{equation}
\bm{\lambda}(D^{(1)})\prec \sum_j p_j\bm{\lambda}(D^{(1)}_j)\label{Dj}
\end{equation} 
for the extended densities. We then say that an  operation not conserving fermion number is ONG if it admits a set of Kraus operators ${\cal K}_j$ such that \eqref{Dj} is satisfied. Prop.\ 1 remains then  valid for general ONG operations replacing \eqref{2} by \eqref{Dm}. All previous properties satisfied by  the entropies  \eqref{E} extend to entropies  \begin{equation}
E(|\Psi\rangle)=S(D^{(1)})\,,\label{ED}
\end{equation}
with ${\rm Tr}f(D^{(1)})={\rm Tr}f(\rho^{(1)})+{\rm Tr}f(\mathbbm{1}-\rho^{(1)})$. In particular,  Eq.\  \eqref{4} becomes just the von Neumann entropy of $D^{(1)}$.  Maximally one-body entangled states are now those leading to
\begin{equation}
D^{(1)}=\mathbbm{1}_{2n}/2\,,
\end{equation}
which will maximize all entropies \eqref{ED}. Examples are previous GHZ-like states \eqref{FGZ} in half-filled SP spaces ($m=2$, $N=n/2\geq 1$). 
 
Theorem 1 then implies the  following result for the basic measurement having $c_k$ and $c^\dag_k$ as  operators: 
\begin{corol}
A measurement on single-particle mode $k$ described by the operators $c_k$ and $c^\dag_k$, which satisfy  $c^\dag_kc_k+c_k c^\dag_k=\mathbbm{1}$, is a ONG operation: 
\begin{equation}
\bm{\lambda}(D^{(1)})\prec p_k\bm{\lambda}(D^{(1)}_k)+p_{\bar{k}}\bm{\lambda}(D^{(1)}_{\bar{k}})\,.\label{leD}
\end{equation}
\end{corol}
Here $p_k=\langle c^\dag_k c_k\rangle=1-p_{\bar{k}}$, and $D^{(1)}$ and $D^{(1)}_{k(\bar{k})}$ are the extended SPDMs determined by $\rho$, $\rho_k=c_k\rho c_k^\dag/p_k$ and $\rho_{\bar{k}}=c^\dag_k\rho c_k/p_{\bar{k}}$. Since these extended densities $D^{(1)}_{k(\bar{k})}$ have clearly the {\it same} spectrum as those obtained from ${\cal P}_k\rho{\cal P}_k/p_k$ and ${\cal P}_{\bar{k}}\rho{\cal P}_{\bar{k}}/p_{\bar{k}}$, Eq.\ \eqref{leD} directly follows from Theorem 1 and Eq.\  \eqref{C2}. On the other hand, previous basic occupation measurement \eqref{Pk} is just the composition of this measurement with itself (see Appendix \ref{C}).

This extension then enables us to consider the addition of free ancillas (SDs of arbitrary $N$) as a free operation, as it  will not alter the  spectrum of the extended SPDM $D^{(1)}$ in the full SP space. We remark, however,  that general ``active'' FLO operations which do not conserve the fermion number  (for instance, a Bogoliubov transformation) may increase the one-body entanglement determined by $\rho^{(1)}$. While we will not discuss these operations here,  we mention that if they are also regarded as free one should consider instead the generalized one-body entanglement, determined by the mixedness of the full quasiparticle DM \cite{GR.15}, as the associated resource (see also Appendix \ref{C}). 

\subsection{One-body entanglement as a resource\label{B4}}

The identification of number conserving FLO operations as ONG implies  that they map SDs onto SDs, as verified in App.\  \ref{B}. It  has been noted  \cite{DT.05,MCT.13} that this fact ultimately explains why the pure state FLO computation model can be efficiently simulated classically, as matrix elements of free unitaries and outcome probabilities of free measurements can be reduced to overlaps $\langle \Psi|\Phi\rangle$ between SDs, which can be computed in polynomial time through a determinant \cite{TD.02}.

In contrast, the simultaneous measurement of the occupancy of two SP modes $k$ and $k'$, described by operators $\{{\cal M}_0={\cal P}_{\bar k}{\cal P}_{\bar k'},{\cal M}_1={\cal P}_{\bar k}{\cal P}_{k'}+{\cal P}_{k}{\cal P}_{\bar k'}, {\cal M}_2={\cal P}_{k}{\cal P}_{k'}\}$, is not free since ${\cal M}_1$ can map a SD onto a state with Slater number 2 \cite{DT.05},   i.e.\ a one-body entangled state. A similar measurement with $m$ such outcomes may return a state with an exponentially large ($2^m$) Slater number \cite{DT.05}, the expectation values of which would be hard to evaluate classically. In  \cite{BD.04} this operation is identified with a charge detection measurement in a system of free electrons, showing that it is possible to build a controlled-NOT gate with just beamsplitters, spin rotations and charge detectors. The extended set of FLO plus charge detection operations then enables quantum computation. If the computational power of this model is to be linked to the presence of a quantum resource, the ensuing free states and operations would be $\cal S$ and $\cal O$, respectively, and the results derived here entail  that one-body entanglement would be an associated resource. 

One-body entanglement can also be considered as a resource for   mode entanglement. In particular,  mode entanglement   with {\it definite particle number $N$ or definite number parity $e^{i\pi N}$ at each component  requires one-body entanglement}. A first example was seen with the normal form \eqref{sfn2} for a general two-fermion state \cite{SC.01,SL.01,ES.02}, where the entanglement between the modes $k$ ($A$) and $\bar{k}$ ($B$), containing each one fermion, is directly linked to one-body entanglement: The entanglement entropy $E(A,B)=S(\rho_A)=S(\rho_B)$ of this partition is just 
\begin{equation}
E(A,B)=\sum_\nu f(\lambda_\nu)=\tfrac{1}{2}E(|\Psi\rangle)
\end{equation}
for any entropy $S(\rho)={\rm Tr}\,f(\rho)$,   where $E(|\Psi\rangle)=S(\rho^{(1)})$ is the corresponding one-body entanglement entropy \eqref{E} (as $\langle c^\dag_\nu c_{\nu'}\rangle=\langle c^\dag_{\bar{\nu}} c_{\bar{\nu}'}\rangle=\lambda_\nu\delta_{\nu\nu'}$, $\langle c^\dag_\nu c_{\bar{\nu}'}\rangle=0$). In particular any one-body entangled state of two fermions in a SP space of dimension 4 can be seen as an entangled state of two distinguishable qubits, allowing then the realization of tasks like quantum teleportation \cite{GR.17}. In this case the one-body entanglement entropy also provides a lower bound to any bipartite mode entanglement entropy \cite{GR.17}.

For a general pure fermionic state $|\Psi\rangle$ (with definite particle number or number parity) we now show that one-body entanglement  {\it is always  required} in order to have bipartite mode entanglement entropy $E(A,B)>0$ with definite particle number, or in general definite number parity, at each side $A$ and $B$: In such a case, and assuming sides $A$ and $B$ correspond to orthogonal subspaces ${\cal H}_A$ and ${\cal H}_B$ of the SP space ${\cal H}={\cal H}_A\oplus{\cal H}_B$, the ensuing SPDM takes the block-diagonal form $\rho^{(1)}=\rho_A^{(1)}\oplus\rho_B^{(1)}$, i.e., 
\begin{equation}
\rho^{(1)}=\begin{pmatrix} \rho^{(1)}_A&0\\0&\rho^{(1)}_B\end{pmatrix}
\end{equation}
since for  any $k_A\in{\cal H}_A$, $k_B\in{\cal H}_B$, $c^\dag_{k_A}c_{k_B}$ connects states with different number parity at each side and hence $\langle c^\dag_{k_A}c_{k_B}\rangle=0$ in such a state. Then, if $|\Psi\rangle$ is a SD, $(\rho^{(1)})^2=\rho^{(1)}$, implying  $(\rho^{(1)}_{A(B)})^2=\rho^{(1)}_{A(B)}$, i.e.\ the state at each side must be a SD (a pure state) and no $A$--$B$ entanglement is directly present. For instance, a single fermion state $\frac{1}{\sqrt{2}}(c^\dag_{k_A}+c^\dag_{k_B})|0\rangle$ implies entanglement between $A$ and $B$  but at the expense of involving zero and one fermion at each side, i.e., no definite local number parity. In contrast,  $|\Psi\rangle=\frac{1}{\sqrt{2}}(c^\dag_{k_A}c^\dag_{k_B}+c^\dag_{k'_A}c^\dag_{k'_B})|0\rangle$ leads to entanglement between $A$ and $B$ with definite fermion number (and hence number parity) at each side, but it is not a SD, i.e., it has nonzero one-body entanglement.  
 
Expanding the state in a SD basis as $|\Psi\rangle=\sum_{\mu,\nu} C_{\mu\nu}A^\dag_\mu B^\dag_{\nu}|0\rangle$, where $A^\dag_\mu=\prod_k(c^\dag_{k_A})^{n_{k_\mu}}$ and $B^\dag_\nu=\prod_k(c^\dag_{k_B})^{n_{k_\nu}}$  involve creation operators just on ${\cal H}_A$ and ${\cal H}_B$ respectively (with $n_{k_{\mu(\nu)}}=0,1$ and $\mu$ and $\nu$ labeling all possible sets of occupation numbers, such that  $\langle 0|A_\mu A^\dag_{\mu'}|0\rangle=\delta_{\mu\mu'}$, $\langle 0|B_\nu B^\dag_{\nu'}|0\rangle=\delta_{\nu\nu'}$), states with definite number parity at each side correspond to $(-1)^{\sum_k n_{k_\mu}}$ and  $(-1)^{\sum_k n_{k_\nu}}$ fixed for all $\mu$ and $\nu$ with $C_{\mu\nu}\neq 0$. The reduced DM of side $A$ is $\rho_A=\sum_{\mu,\mu'}(CC^\dag)_{\mu\mu'}A^\dag_{\mu}|0\rangle\langle0|A_{\mu'}$ (and similarly for $\rho_B$; see Appendix \ref{A}), and there is entanglement between $A$ and $B$ whenever $\rho_A$ has rank $\geq 2$, i.e., rank$(C)\geq 2$. In such a case, previous argument shows that such $|\Psi\rangle$ cannot be a SD if the fermion number or number parity is fixed at each side, i.e. $N_A$ or $e^{i\pi N_A}$ fixed. 

Due to the fermionic number parity superselection rule \cite{FL.13,Fr.16}, fixed number parity at each side is required in order to be able to form arbitrary superpositions, i.e., arbitrary unitary transformations   of the local states, and hence to have entanglement fully equivalent to the distinguishable case. Fixed particle number at each side may be also required if the particle number or charge superselection rule applies for the fermions considered.  

The extension to multipartite mode entanglement is straightforward: For a decomposition ${\cal H}=\bigoplus_i {\cal H}_i$ of the SP space into orthogonal subspaces ${\cal H}_i$, and for component $i$ associated to   subspace ${\cal H}_i$,  all elements of $\rho^{(1)}$ connecting different components $i$ and $j$  will vanish if each component is to have definite fermion number or number parity in a state $|\Psi\rangle$: $\langle c^\dag_{k_i}c_{k_j}\rangle=0$ $\forall$ $k_i,k_j$  if $i\neq j$. Thus, 
\begin{equation}
\rho^{(1)}=\bigoplus_i\rho_i^{(1)}\,.\label{47}
\end{equation}
If $|\Psi\rangle$ were a SD, each $\rho^{(1)}_i$ should  then satisfy $(\rho^{(1)}_i)^2=\rho^{(1)}_i$ and hence each subsystem would be in a pure SD state, implying no entanglement between them. A one-body entangled state is then required. 

And when all previous subsystems contain just one fermion, one-body entanglement is directly linked to standard multipartite entanglement.  Consider an $N$-partite system with Hilbert space ${\cal L}=\bigotimes_{i=1}^{N}{\cal L}_i$, where ${\cal L}_i$ is the Hilbert space of the $i$-th (distinguishable) constituent. Consider also an $N$-fermion system with  SP space ${\cal H}=\bigoplus_{i=1}^N {\cal H}_i$, such that ${\rm dim}\,{\cal H}_i={\rm dim}\,{\cal L}_i$. This  enables the definition of an isomorphism $\Theta_i:{\cal L}_i\rightarrow {\cal H}_i$ 
between these two spaces: Any pure separable state in $\cal L$, $|S\rangle_{\cal L}=\bigotimes_{i=1}^{N}|\phi_i\rangle$, with $|\phi_i\rangle\in{\cal L}_i$, can be mapped to a SD 
\begin{equation}
|\Psi\rangle=[\prod_{i=1}^N c^\dagger_{i,\phi}]|0\rangle,=\Theta(|S\rangle_{\cal L}),\label{ib1}
\end{equation}
where $c^\dagger_{i, \phi}$ creates a fermion in the state $\Theta_i(|\phi_i\rangle)\in{\cal H}_i$. 

The map $\Theta:{\cal L}\rightarrow{\cal F}$ is an isomorphism between $\cal L$ and the subspace of $\cal F$ determined by the fermion states having occupation number $N_i=1$ in ${\cal H}_i\subset{\cal H}$. Hence any pure state $|\Phi\rangle_{\cal L}$ of the multipartite system is mapped to a state $|\Psi\rangle=\Theta(|\Phi\rangle_{\cal L})$ in $\cal F$. Since there is a fixed fermion  number (1) in each constituent, the SPDM $\rho^{(1)}$ determined by $|\Psi\rangle$ will satisfy Eq.\ \eqref{47}, implying here 
\begin{equation}
\rho^{(1)}=\bigoplus_i \rho_i,  \label{ibd}
\end{equation} 
where the elements of the matrix $\rho_i$ are those of the reduced state of the $i$-th subsystem associated to $|\Phi\rangle_{\cal L}$. Thus, one-body entanglement monotones become  $E(|\Psi\rangle)={\rm Tr}\,f(\rho^{(1)})=\sum_i {\rm Tr}\,f(\rho_i)$, being then equivalent to  the multipartite version of the {\it linear entropy of entanglement} \cite{BK.04, WD.13,MW.02,BM.08} and constituting monotones for the multipartite entanglement of the tensor product representation.

This link between one-body entanglement and multipartite entanglement is not a surprise if we recall that performing local operations on the multipartite system cannot increase, on average, the mixedness of the local eigenvalues. Relation \eqref{ibd} then implies that these operations cannot increase the mixedness of the SPDM associated to the fermionic representation, in agreement with Eq.\ (2). And any local unitary in the tensor product representation can be implemented as  a one-body unitary in the fermion system, while  local projective measurements can be performed as occupation measurements, both  FLO operations  which are one-body entanglement nongenerating. The overlap between one-body entanglement and multipartite entanglement described here reinforces the idea that the first could be the resource behind the computational power of the quantum computation model described in \cite{BD.04}, since it  consists of a mapping from qubits to fermions just like the map $\Theta$ defined above.

\section{Conclusions}

We have shown that one-body entanglement,  a measure of the deviation of a pure fermionic state from a SD determined by the mixedness of the SPDM $\rho^{(1)}$, can be considered as a quantum resource. We have first provided a basis-independent bipartite like formulation of one-body entanglement in general $N$-fermion states, which relates it with the correlation between one and ($N-1$)-body observables and is analogous to that of systems of distinguishable components. Such formulation  leads to a Schmidt-like decomposition of the state, which contains the common eigenvalues of the one and $(N-1)$-body DMs. 

We have then identified the class of one-body entanglement nongenerating operations through the  rigorous majorization relation \eqref{avmaj}. And we have  shown in Theorem 1 that single mode occupation measurements satisfy indeed this relation, implying they will not increase, on average,  the one-body entanglement here defined. The ensuing theory then has SDs as free states and number conserving FLO operations as free operations.  We have also considered  in corollary 1 and 2 more general occupation measurements, showing they are also ONG. Connections with mode entanglement and multipartite entanglement have also been discussed, showing in particular that one-body entanglement is required for  entanglement with fixed fermion number or number parity at each subsystem. Present results provide the basis for a consistent resource theory associated to quantum  correlations beyond antisymmetrization in fermionic systems, which should play a fundamental role in fermionic protocols beyond the FLO model. 

\acknowledgments
The authors acknowledge support from CONICET (N.G. and M.DT.) and CIC (R.R.) of Argentina. This work was supported by CONICET PIP Grant No. 11220150100732.

\appendix

\section{Proof of inequality \eqref{ineq}  and Corollary 1\label{A}} 

\begin{proof}
We first prove  Eq.\ \eqref{ineq} for the lowest eigenvalue of $\rho^{(1)}_{{\cal S}'_m}$, where ${\cal S}'_m$ is the subspace  containing the measured SP state $|k\rangle$ and the first $m$ eigenstates of $\rho^{(1)}$. We  write the lowest eigenstate  of $\rho^{(1)}_{{\cal S}'_m}$ as $\alpha |k\rangle+\beta |k'\rangle$, with $|k'\rangle\in{\cal S}'_m$ orthogonal to  $|k\rangle$, such that $\lambda_{m+1}(\rho^{(1)}_{{\cal S}'_m})$  is the smallest eigenvalue $\lambda_-$ of the $2\times 2$ matrix
\begin{equation}
\rho^{(1)}_{kk'}=\begin{pmatrix}
\langle c^\dagger_k c_k\rangle&&\langle c^\dagger_{k'} c_k\rangle\\
\langle c^\dagger_k c_{k'}\rangle&&\langle c^\dagger_{k'} c_{k'}\rangle
\end{pmatrix}\label{mkp}\,.
\end{equation}
Setting $\langle c^\dagger_{k} c_{k}\rangle=p_k$, its eigenvalues are  
\begin{equation} 
\lambda_\pm ={\textstyle\frac{p_k+p_{k'}}{2}\pm\sqrt{\frac{(p_k-p_{k'})^2}{4}+|\langle c^\dag_{k'} c_{k}\rangle|^2}}\label{ineqm1}\,.
\end{equation}
Writing again ${\cal P}_k=c^\dag_k c_k$, ${\cal P}_{\bar{k}}=c_kc^\dagger_k=\mathbbm{1}-{\cal P}_k$, a general state $|\Psi\rangle$ can be expanded as 
\begin{equation}
{\textstyle|\Psi\rangle=\sum_{\mu,\mu'}{\cal P}_\mu {\cal P}_{\mu'}|\Psi\rangle =\sum_{\mu,\mu'}\sqrt{p_{\mu\mu'}}|\Psi_{\mu\mu'}\rangle\,,}
\end{equation} 
where $\mu=k,\bar{k}$; $\mu'=k',\bar{k}'$ (so that $\sum_{\mu,\mu'}{\cal P}_\mu{\cal P}_{\mu'}=\mathbbm{1}$)  and $|\Psi_{\mu\mu'}\rangle={\cal P}_{\mu}{\cal P}_{\mu'}|\Psi\rangle/\sqrt{p_{\mu\mu'}}$, with $p_{\mu\mu'}=\langle \Psi|{\cal P}_\mu{\cal P}_{\mu'}|\Psi\rangle$, are states with definite occupation of SP states $|k\rangle$ and $|k'\rangle$. We then have $p_k=p_{kk'}+p_{k\bar{k'}}$, $p_{k'}= p_{kk'}+p_{\bar{k}k'}$ and $\langle c^\dag_{k'}c_{k}\rangle= r\sqrt{p_{\bar{k}k'}p_{k\bar{k'}}}$, with $|r|\leq 1$. Thus, $|\langle c^\dag_{k'}c_{k}\rangle|^2\leq p_{\bar{k}k'}p_{k\bar{k}'}$, and \eqref{ineqm1}   implies $\lambda_-\geq \frac{p_k+p_{k'}}{2}-\frac{p_{k\bar{k}'}+p_{\bar{k}k'}}{2}$, i.e., 
\begin{equation}
\lambda_-\geq  p_{kk'}\geq p_k \lambda_{m+1}(\rho^{(1)}_{k{\cal S}'_m})+p_{\bar{k}}\lambda_{m+1}(\rho^{(1)}_{\bar{k}{\cal S}'_m})\,,\label{ineqm2}
\end{equation}
since $p_k \lambda_{m+1}(\rho^{(1)}_{k{\cal S}'_m})\leq p_{kk'}$ and $\lambda_{m+1}(\rho^{(1)}_{\bar{k}{\cal S}'_m})=0$ (as the state $k$ is empty in $\rho^{(1)}_{\bar{k}{\cal S}'_m}$). This implies  Eq.\ \eqref{ineq} since $\lambda_{m+1}(\rho^{(1)}_{{\cal S}'_m})=\lambda_-$.
\end{proof}
We also mention that by considering the largest eigenvalue $\lambda_+$ in \eqref{ineqm1} of a similar $2\times 2$ block, such that the first eigenstate (largest eigenvalue) of $\rho^{(1)}$ is spanned by $|k\rangle$ and $|k'\rangle$, it is verified, using again $|\langle c^\dag_{k'}c_{k}\rangle|^2\leq p_{\bar{k}k'}p_{k\bar{k}'}$, that $\lambda_1(\rho^{(1)})=\lambda_+\leq p_k+p_{\bar{k}k'}\leq p_k\lambda_1(\rho^{(1)}_k)+p_{\bar{k}}\lambda_1(\rho^{(1)}_{\bar{k}})$, since $\lambda_1(\rho^{(1)}_k)=1$ and $p_{\bar{k}k'}\leq p_{\bar{k}}\lambda_1(\rho^{(1)}_{\bar{k}})$, which is the first ($m=1$) inequality in Eq.\ \eqref{major1}.

The present proof of Theorem 1 also holds for general mixed fermionic states $\rho$ (assumed  to commute with the fermion number $\hat{N}=\sum_k c^\dag_k c_k$ or in general the  number parity $e^{i\pi \hat{N}}$) since they can always be purified, i.e.\  considered as reduced states $\rho_{\cal S}$ of a pure fermionic state
\begin{equation}
|\Psi\rangle=\sum_{\mu,\nu}C_{\mu\nu}A^\dag_\mu B^\dag_\nu|0\rangle\,, \label{psip}
\end{equation}
of definite number parity. Here  $A^\dag_{\mu}$ and $B^\dag_\nu$  contain  creation operators in ${\cal S}$ and in an orthogonal SP space ${\cal S}_\perp$ respectively, satisfying $\langle 0|A_{\mu'}A^\dag_{\mu}|0\rangle=\delta_{\mu\mu'}$, $\langle 0|B_{\nu'}B^\dag_{\nu}|0\rangle=\delta_{\nu\nu'}$, $\langle 0|B_{\nu'}A^\dag_{\nu}|0\rangle=0$. We may assume, for instance, that $\{A^\dag_\mu B^\dag_\nu|0\rangle\}$ is a complete set of orthogonal SDs. Then $\rho_{\cal S}={\rm Tr}_{{\cal S}_{\perp}}|\Psi\rangle\langle\Psi|=\sum_{\mu,\mu'}(C C^\dag)_{\mu\mu'}|\mu\rangle\langle\mu'|$, with $|\mu\rangle=A^\dag_\mu|0\rangle$ a SD in ${\cal S}$, satisfies $\langle\Psi|O_S|\Psi\rangle={\rm Tr}\rho_{\cal S}O_{\cal S}$ $\forall$ operators   $O_S$ containing creation and annihilation operators of SP states $\in{\cal S}$. Given then an arbitrary mixed fermionic state $\rho_{\cal S}$ Eq.\ \eqref{psip} is a purification of $\rho_{\cal S}$ for any matrix $C$ satisfying $(CC^\dag)_{\mu\mu'}= \langle\mu|\rho_{\cal S}|\mu'\rangle$ (requires dim$\,{\cal S}_\perp\geq\,$dim$\,{\cal S}$).
  
Let us now prove Corollary 1, i.e.\ the extension of Theorem 1 to the more general occupancy measurement operators of Eq.\ \eqref{Mk}. In terms of the states \eqref{pek}, the ensuing postmeasurement states $|\Psi'_{k(\bar{k})}\rangle\propto M_{k(\bar{k})}|\Psi\rangle$ are 
\begin{eqnarray}
|\Psi'_k\rangle&=&\left(\alpha\sqrt{p_k}\, |\Psi_k\rangle+\beta\sqrt{p_{\bar k}}\, |\Psi_{\bar{k}}\rangle\right)/\sqrt{p'_k},\label{estk}\\
|\Psi'_{\bar k}\rangle&=&\left(\gamma\sqrt{p_k}\, |\Psi_k\rangle+\delta\sqrt{p_{\bar k}}\, |\Psi_{\bar{k}}\rangle\right)/\sqrt{p'_{\bar k}},\label{estkb}
\end{eqnarray}
where $|\alpha|^2+|\gamma|^2=1$, $|\beta|^2+|\delta|^2=1$ and 
\begin{eqnarray}
p'_k&=&p_k|\alpha|^2+p_{\bar{k}}|\beta|^2\,,\;\;\;
p'_{\bar k}=p_k|\gamma|^2+p_{\bar{k}}|\delta|^2\,,
\end{eqnarray}
with $p'_k+p'_{\bar{k}}=1$. We have to prove 
\begin{equation}
\bm\lambda(\rho^{(1)})\prec p'_k\bm\lambda({\rho'}_k^{(1)})  + p'_{\bar k}
\bm\lambda({\rho'}^{(1)}_{\bar k})\,,\label{8}
\end{equation}
where ${\rho'}^{(1)}_{k(\bar{k})}$ are now the SPDM's determined by  the states \eqref{estk}--\eqref{estkb}. The generalization of Eq.\ \eqref{Tr},  
\begin{equation}
{\rm Tr}\,\rho_{\cal S}^{(1)}={\rm Tr}\,[p'_k{\rho'}^{(1)}_{k{\cal S}}+p'_{\bar{k}}{\rho'}^{(1)}_{\bar{k}{\cal S}}]\,,\label{9}
\end{equation} 
still holds for any subspace ${\cal S}$ either orthogonal to or containing the SP state $|k\rangle$, as $[M_{k(\bar{k})},c^\dag_{k'} c_{k'}]=0$ for both $k'=k$ or $k'$ orthogonal to $k$ [see comment below Eq.\ \eqref{rhort}]. Proceeding in the same way and using previous notation, we see that $p'_k\lambda_{m+1}({\rho'}^{(1)}_{k {\cal S}'_m})$ is less than or equal to the smallest eigenvalue $\lambda_{k-}$ of the $2\times 2$ matrix 
\begin{equation}
\begin{pmatrix}
|\alpha|^2p_k&&\alpha\beta^*\langle c^\dag_{k'}c_k\rangle\\
\alpha^*\beta\langle c^\dag_{k}c_{k'}\rangle&& |\alpha|^2 p_{kk'}+|\beta|^2p_{\bar k k'}
\end{pmatrix},
\end{equation}
while $p'_{\bar{k}}\lambda_{m+1}({\rho'}^{(1)}_{\bar{k} {\cal S}'_m})$ is less than or equal to the smallest eigenvalue  $\lambda_{\bar{k}-}$ of a similar matrix with $\alpha\rightarrow \gamma$, $\beta\rightarrow\delta$. It is then  straightforward to prove, using  Eq.\ \eqref{ineqm1} for $\lambda_-$, that  $\lambda_{m+1}(\rho^{(1)}_{{\cal S}'_m})=\lambda_-\geq \lambda_{k-}+\lambda_{\bar{k}-}$, since 
\begin{eqnarray}
\lambda_--\lambda_{k-}-\lambda_{\bar{k}-}&=& {\textstyle\sqrt{\frac{(|\alpha|^2 p_{k\bar{k'}}-|\beta|^2 p_{\bar{k}k'})^2}{4}+|\alpha\beta\langle c^\dag_{k'}c_k\rangle|^2}}\nonumber\\
&&+{\textstyle\sqrt{\frac{(|\gamma|^2 p_{k\bar{k'}}-|\delta|^2 p_{\bar{k}k'})^2}{4}+|\gamma\delta\langle c^\dag_{k'}c_k\rangle|^2}}\nonumber\\&&
{\textstyle-\sqrt{\frac{(p_{k\bar{k'}}-p_{\bar{k}k'})^2}{4}+|\langle c^\dag_{k'}c_k\rangle|^2}}\nonumber\\&\geq &0\,,
\end{eqnarray}
with equality for $|r|=1$ ($|\langle c^\dag_{k'}c_k\rangle|=\sqrt{p_{\bar{k}k'}p_{k\bar{k'}}}$) or $|\alpha|=|\beta|$. Then the $m^{\rm th}$ inequality in \eqref{8}, 
\begin{equation}
\sum_{\nu=1}^m\lambda_\nu(\rho^{(1)})\leq \sum_{\nu=1}^m p'_k\lambda_\nu({\rho'}^{(1)}_k)+ p'_{\bar{k}}\lambda_\nu({\rho'}^{(1)}_{\bar k})\,,
\end{equation} 
follows due to \eqref{9} and the previously used Ky Fan inequality.  Eq.\ \eqref{8} also holds within any subspace ${\cal S}$ containing (or orthogonal to) the SP state $|k\rangle$.$\!\!$
\qed

\section{Occupation measurements on free states\label{B}} 

For a one-body entanglement nongenerating operation, Eq.\ \eqref{avmaj}  implies that the Kraus operators ${\cal K}_j$ satisfying it should convert free states onto free states, i.e.\ SDs onto SDs. For the occupation measurements of Eq.\ \eqref{Pk} (Theorem 1), Eq.\ \eqref{Mk} (Corollary 1), and Corollary 2, this property can be easily verified. Let  
\begin{equation}
|\Psi\rangle=(\prod_{\nu=1}^N c^\dag_\nu)|0\rangle\,,
\end{equation}
be a general SD for $N$ fermions, with $\{c_\nu,c^\dag_{\nu'}\}=\delta_{\nu\nu'}$. A general fermion creation operator  $c^\dag_k=\sum_{\nu=1}^n \alpha_\nu c^\dag_\nu$, with $\alpha_\nu=\{c_\nu,c^\dag_k\}$ and $\{c_k,c^\dag_k\}=\sum_{\nu}|\alpha_\nu|^2=1$, can be written as  
\begin{equation}
c^\dag_k=\sqrt{p_k} c^\dag_{k_\parallel}+\sqrt{p_{\bar k}}c^\dag_{k_\perp}\,,
\end{equation}  
where $\sqrt{p}_k c^\dag_{k_{\parallel}}=\sum_{\nu\leq N} \alpha_\nu c^\dag_\nu$ is the component in the subspace occupied in $|\Psi\rangle$, with $p_k=\sum_{\nu\leq N} |\alpha_\nu|^2=\langle \Psi|c^\dag_k c_k|\Psi\rangle$ the occupation probability of SP state $k$ and $c^\dag_{k_\parallel}|\Psi\rangle=0$, while $\sqrt{p}_{\bar{k}} c^\dag_{k_{\perp}}=\sum_{\nu> N} \alpha_\nu c^\dag_\nu$ is the orthogonal complement, with  $p_{\bar{k}}=\sum_{\nu>N}|\alpha_\nu|^2=1-p_k$ and $c_{k_\perp}|\Psi\rangle=0$. If $p_k>0$, through a unitary transformation of the  $c^\dag_\nu$ for $\nu\leq N$, they  can be chosen such that $c^\dag_{k\parallel}=c^\dag_{\nu=N}$. Hence, for the measurement operators of corollary 2, we see that 
\begin{eqnarray}
c_k|\Psi\rangle&=&\sqrt{p_k} c_{k_\parallel}|\Psi\rangle, \;\; c^\dag_k|\Psi\rangle=\sqrt{p_{\bar k}}c^\dag_{k_\perp}|\Psi\rangle\,,\label{B3}
\end{eqnarray}
are clearly orthogonal SDs. For the measurement operators of Theorem 1,  Eq.\ \eqref{B3}  implies  
\begin{eqnarray}
c^\dag_k c_k|\Psi\rangle&=&\sqrt{p_k} c^\dag_k c_{k_\parallel}|\Psi\rangle\,,\label{B22}\\ c_kc^\dag_k|\Psi\rangle&=&\sqrt{p_{\bar{k}}} c_k c^\dag_{k_\perp}|\Psi\rangle= \sqrt{p_{\bar k}} c^\dag_{k'}c_{k_\parallel}|\Psi\rangle\,,\label{B23}
\end{eqnarray} 
where $c^\dag_{k'}=\sqrt{p_{\bar k}}c^\dag_{k_{\parallel}}-\sqrt{p_k}c^\dag_{k_\perp}$, which are also orthogonal SD's ($\{c_k,c^\dag_{k'}\}=0$). And in the case of the generalized measurement  based on the operators \eqref{Mk}, we  see from \eqref{B22}--\eqref{B23} that  
\begin{eqnarray}
M_k|\Psi\rangle&=& (\alpha\sqrt{p_k}c^\dag_k+\beta\sqrt{p_{\bar{k}}}c^\dag_{k'}) c_{k_{\parallel}}|\Psi\rangle\\
M_{\bar{k}}|\Psi\rangle&=& (\gamma\sqrt{p_k}c^\dag_k+\delta\sqrt{p_{\bar{k}}}c^\dag_{k'}) c_{k_{\parallel}}|\Psi\rangle
\end{eqnarray}
are as well SDs, not necessarily orthogonal.

\section{Comparing one-body entanglement of states with different particle number\label{C}}

Given two pure fermionic states $|\Psi\rangle$ and $|\Phi\rangle$ with the same fermion number $N$, such that their associated SPDMs have the same trace  ${\rm Tr}\,\rho^{(1)}_\Psi={\rm Tr}\,\rho^{(1)}_\Phi=N$, $|\Psi\rangle$ is  considered not less entangled than $|\Phi\rangle$ if Eq.\ \eqref{2} ($\bm{\lambda}(\rho^{(1)}_{\Psi})\prec\bm{\lambda}(\rho^{(1)}_{\Phi})$) is satisfied. Here $\bm{\lambda}(\rho^{(1)})$ denotes the spectrum of $\rho^{(1)}$ sorted in decreasing order. It can be shown that
\begin{equation}
\bm{\lambda}(\rho^{(1)}_{\Psi})\prec\bm{\lambda}(\rho^{(1)}_{\Phi})\;\Longrightarrow\;  \bm{\lambda}(D^{(1)}_{\Psi})\prec\bm{\lambda}(D^{(1)}_{\Phi})\,,\label{C2}
\end{equation}
for the extended  DM  defined in  Eq.\ \eqref{D1}, the spectrum of which is $(\bm{\lambda},1-\bm{\lambda})$. Eq.\ \eqref{C2}  follows from the  straightforward properties  (see for instance \cite{MOA.11,Bh.97})\\
i) $\bm{\lambda}\prec\bm{\lambda}'$ 
$\Longrightarrow$ $1-\bm{\lambda}\prec 1-\bm{\lambda}'$,\\
ii) $\bm{\lambda}\prec\bm{\lambda}'$ and   $\bm{\mu}\prec\bm{\mu}'$ $\Longrightarrow$ $(\bm{\lambda},\bm{\mu})\prec(\bm{\lambda}',\bm{\mu}'),$\\
where $\bm{\lambda},\bm{\lambda}', \bm{\mu},\bm{\mu}'\in\mathbb{R}^n$ are sorted in decreasing order and $(\bm{\lambda},\bm{\mu})\in\mathbb{R}^{2n}$ denotes the sorted vector resulting from  the union of $\bm{\lambda}$ and $\bm{\mu}$. The converse relation in \eqref{C2} does not hold. These properties also entail that the majorization relation \eqref{avmaj} implies 
\begin{equation}
\bm{\lambda}(D^{(1)})\prec \sum_j p_j\bm{\lambda}(D^{(1)}_j),\label{C3}
\end{equation}
for the corresponding extended densities. 

The advantage of Eqs.\ \eqref{C2}--\eqref{C3} is that within a fixed SP space the extended vectors can always be compared through  majorization, regardless of the particle number $N$, since ${\rm Tr}\,D^{(1)}=n={\rm dim}\,{\cal H}$ is fixed by the SP space dimension. For two states $|\Psi\rangle$ and $|\Phi\rangle$ with definite but not necessarily coincident fermion number, we then say that $|\Psi\rangle$ is not less entangled than $|\Phi\rangle$ if Eq.\ \eqref{C2} holds. In particular, it is clear that —up to a permutation— the same eigenvalue vector $\bm{\lambda}(D^{(1)})$ is assigned to all SD states in the Fock space of the system, irrespective of $N$, so that they are all least entangled states. 

The extension of Definition 1 to general ONG operations, not necessarily conserving the particle number, is now straightforward: a quantum operation is ONG if it admits a set of Kraus operators $\{{\cal K}_j\}$ satisfying Eq.\ \eqref{C3}   $\forall$ $\rho$, with $D^{(1)}$  and $D^{(1)}_j$ the extended SPDMs determined by $\rho$ and $\rho_j$ respectively. This extension allows us to consider  operations  such as that of Corollary 2, with Kraus operators  $c_k$ and $c^\dag_k$. The extended SPDMs $D^{(1)}$  of the postmeasurement states $\frac{1}{\sqrt{p_k}} c_k|\Psi\rangle$ and $\frac{1}{\sqrt{p_{\bar{k}}}} c^\dag_k|\Psi\rangle$ have clearly the same spectrum as those obtained from $\frac{1}{\sqrt{p_k}} c^\dag_kc_k|\Psi\rangle$ and $\frac{1}{\sqrt{p_{\bar{k}}}} c_kc^\dag_k|\Psi\rangle$ —up to a permutation—, with the same probabilities, such that theorem 1 directly implies corollary 2. In fact, the  number conserving occupation measurement is just  a composition of the former with  itself, as $(c_k,c^\dag_k)\circ(c_k,c^\dag_k)=(c_kc^\dag_k,c^\dag_k c_k)$. 
 
We can therefore embed the fermion number preserving resource theory  within a more general theory in which the set of free states is the convex hull of SD states —of all possible particle number— and where the free operations are one-body unitaries and  operations based on the $\{c_k,c^\dag_k\}$ measurement mapping SDs onto SDs. Any SD and hence any free state can be prepared from an arbitrary state $\rho$ by means of free operations only, i.e., by applying one-body unitaries and successive $\{c_k,c^\dag_k\}$ measurements with postselection.  Since the starting state is arbitrary, any free state in this theory can be converted into any other free state by free operations. Allowing the particle number to vary implies that appending free states is also a free operation, since this will not alter the spectrum of the associated $D^{(1)}$ in the full SP space. 

We have here considered pure states $|\Psi\rangle$ with definite fermion number and operators ${\cal K}_j$ which produce states ${\cal K}_j|\Psi\rangle$ with definite fermion number when applied to such states, suitable for systems where a particle number superselection rule applies. The extension to the case where general fermionic Gaussian states (with no fixed particle number but definite number parity) and active FLO operations are also considered free is straightforward. It requires the consideration of the full extended  $2n\times 2n$ quasiparticle density matrix containing in addition the pair creation and annihilation contractions $\langle c^\dag_k c^\dag_{k'}\rangle$ and $\langle c_{k'} c_k\rangle$, the eigenvalues of which remain invariant under general Bogoliubov transformations. Its mixedness determines a generalized one-body entanglement \cite{GR.15} which vanishes if and only if the state is a SD or a quasiparticle vacuum, i.e., a general pure fermionic Gaussian state.

\end{document}